\documentclass[superscriptaddress,prx,amsmath,amssymb,nofootinbib,twocolumn]{revtex4-2}
\usepackage[T1]{fontenc}
\usepackage{amsthm}
\usepackage{graphicx}
\usepackage[percent]{overpic}
\usepackage{times}
\usepackage{hyperref}
\hypersetup{
colorlinks=true,
linkcolor=red,
citecolor=red,}
\usepackage{textcomp}
\usepackage{enumerate}
\usepackage{multirow}
\usepackage{tabularx}
\usepackage[mathscr]{euscript}
\usepackage{mathtools}
\usepackage{mhchem}
\usepackage[mathscr]{euscript}
\usepackage{dsfont}
\usepackage{amsfonts}
\usepackage[normalem]{ulem}
\usepackage{centernot}
\usepackage{stmaryrd}
\usepackage{longtable}
\usepackage{enumitem,kantlipsum}
\usepackage{array}
\usepackage{wrapfig}
\usepackage{tabu}
\usepackage{graphicx}
\usepackage{placeins}
\usepackage{enumitem}
\usepackage{dutchcal}
\usepackage{tabularx}
\usepackage[export]{adjustbox}
\usepackage{color,soul}
\usepackage{longtable}
\usepackage[dvipsnames]{xcolor}
\usepackage{float}

\usepackage[caption=false]{subfig}
\usepackage{adjustbox}

\usepackage{tikz}
\usetikzlibrary{quantikz}

\newtheorem{definition}{Definition}

\newtheorem{lemma}{Lemma}
\newtheorem{theorem}{Theorem}
\newtheorem*{theorem*}{Theorem}

\theoremstyle{definition}

\newcommand{\ketbra}[1]{ | #1 \rangle\!\langle #1 |}

\newcommand{\one}{\leavevmode\hbox{\small1\normalsize\kern-.33em1}}

\begin{document}

\title{Sheaf-Theoretic Preparation Contextuality}

\author{Tom Williams}
\affiliation{School of Informatics, Quantum Software Lab, University of Edinburgh, United Kingdom}
\author{Mina Doosti}
\affiliation{School of Informatics, Quantum Software Lab, University of Edinburgh, United Kingdom}
\author{Farid Shahandeh}
\affiliation{Department of Computer Science, Royal Holloway, University of London, United Kingdom}

%\date{\today}

\begin{abstract}
We introduce a preparation-dual notion of contextuality, formulated as an obstruction to stochastic extension.
In parallel with the sheaf-theoretic formulation of measurement contextuality, preparation contextuality arises when locally specified preparation statistics cannot be extended to a single global response matrix compatible with all source contexts.
Whereas measurement contextuality concerns the incompatibility of restriction maps (marginalisation), the preparation setting requires stochastic extension of partial conditioning data, which is inherently non-unique.
We identify minimal structural and preparation compatibility conditions on admissible extension matrices and show that they enforce a rigid product form.
This leads to a notion of preparation contextuality in which the absence of any admissible global response representation witnesses contextuality, while preparation compatibility identifies the cases in which this obstruction is nontrivial.
The framework is formulated explicitly in matrix form and illustrated by a quantum-mechanical example exhibiting preparation contextuality.
\end{abstract}

\maketitle

\section{Introduction}
Contextuality has emerged as a central nonclassical feature of quantum theory, capturing the impossibility of explaining certain experimental statistics by a single underlying classical description \cite{Kochen1967,Peres1991,Mermin1993,Cabello1996}.
It plays a foundational role in quantum information, where it underlies nonlocality and Bell inequalities \cite{Bell1964,Fine1982,Brunner2014}.
Various mathematical approaches to contextuality have been developed, including operational formulations \cite{Spekkens2005,Harrigan2010}, logical approaches \cite{Abramsky2012}, sheaf-theoretic formulations \cite{Abramsky2011}, contextuality-by-default \cite{Dzhafarov2017}, and graph- and hypergraph-theoretic approaches \cite{Cabello2014,Acin2015}.

In the standard treatment, contextuality is associated with measurements.
Empirical data are obtained from different measurement contexts, and noncontextuality corresponds to the existence of a single global assignment of outcomes whose restrictions reproduce the observed statistics \cite{Abramsky2011}.
Equivalently, measurement contextuality can be characterised as the failure to extend a family of compatible local distributions to a global one, closely related to Fine's joint-distribution criterion in Bell scenarios \cite{Fine1982}.
In the sheaf-theoretic formulation, contextuality arises as an obstruction to global section existence under canonical restriction by marginalisation.
This perspective is closely related to earlier topos-theoretic work of Butterfield and Isham, who showed that the Kochen--Specker theorem can be reformulated as the non-existence of global sections of a suitable presheaf \cite{Isham1998}.

In this work we develop a complementary, preparation-dual notion of contextuality within the sheaf-theoretic framework.
Existing approaches to preparation contextuality, most notably Spekkens' operational formulation~\cite{Spekkens2005,Kunjwal2015,Schmid2018}, characterise noncontextuality in terms of ontological models constrained by operational equivalences \cite{Harrigan2010}.
While powerful, these approaches do not formulate preparation contextuality as an obstruction to extending local preparation data to a single global object, in the manner that the sheaf-theoretic treatment does for measurement contextuality~\cite{Abramsky2011}.
Our aim is to identify and formalise precisely this extension structure for preparation scenarios.

The key structural distinction is that, although both measurement and preparation contextuality can be formulated as extension problems, the passage from global to local data is fundamentally different.
In the measurement case, local statistics are obtained from global data by canonical marginalisation.
In the preparation case, they arise through stochastic averaging over unobserved preparation degrees of freedom.
The corresponding completion maps are therefore not fixed: there are generally many inequivalent ways to extend a partial preparation specification to a global one.
This non-uniqueness is not merely a technical complication.
In the absence of a canonical completion rule, one must impose structural constraints on admissible extension matrices.
No analogous step is required in the measurement framework, where restriction maps are determined a priori.

We formulate preparation contextuality as an obstruction to stochastic extension.
An empirical model is preparation noncontextual if there exists an admissible extension map together with a single global response matrix reproducing the observed local preparation statistics; otherwise it is preparation contextual.
Admissible extension matrices are defined by two minimal structural requirements: input independence and compositionality.
These conditions rigidly constrain the form of completion, forcing all admissible extension matrices to factor across sources.
Preparation contextuality arises when this global response matrix fails to exist.

Throughout this paper, we adopt an explicit matrix formulation that renders the underlying structure transparent and makes the extension problem a question of stochastic matrix feasibility.
This perspective clarifies the duality between restriction and extension and facilitates comparison with operational and factorisation-based approaches to contextuality.
The general framework is illustrated by a quantum-mechanical example inspired by Pusey–Barrett–Rudolph-type constructions \cite{Pusey2012}, which exhibits preparation contextuality.
The PBR theorem has stimulated a substantial body of work on the status of the quantum state in ontological models \cite{Leifer2013,Barrett2014,Leifer2014}, including analyses of preparation independence and related structural assumptions \cite{Mansfield2014}.

\section{Sheaf-theoretic (measurement) contextuality}
Before turning to preparations, we briefly recall how the sheaf-theoretic formulation of measurement contextuality can be written as an explicit linear-algebraic extension problem \cite{Abramsky2011}.
Subsequent work has developed this framework in several directions, including logical Bell inequalities \cite{Abramsky2012}, cohomological obstructions to contextuality \cite{AbramskyMansfield2011}, and quantitative measures such as the contextual fraction \cite{Abramsky2017}.

Let $X$ be a finite set of measurements, each with outcome set $O$, and let $\mathcal C$ be a cover of $X$.
The elements of $\mathcal C$, called contexts, specify the structure of the scenario under consideration by indicating which joint measurements are admissible.
This structure can be captured by a linear transformation as follows.

Denote by $s_C$ a tuple of outcomes for measurements in a context $C \in \mathcal C$ and call it a \textit{local section}.
The space of all such outcome assignments is $O^{C}$\footnote{Here $O^{C}:=\{\,s:C\to O\,\}$ denotes the set of all functions assigning an outcome in $O$ to each measurement in the context $C$.}.
Similarly, a \textit{global section} $s_X$ is an assignment of outcomes to all measurements in $X$ and their space is $O^X$.
The causal structure of the scenario is then encoded in the \textit{incidence} matrix
$M_{C\mid X}$,
\begin{equation}\label{eq:incidence_matrix}
M_{C\mid X}(s_C \mid s_X) = \mathbf 1[(s_X)|_C = s_C].
\end{equation}
Thus $M_{C\mid X}$ simply ``forgets'' the outcomes of measurements in $X\setminus C$ by summing over all global assignments $s_X$ compatible with a given local assignment $s_C$.
The incidence matrix also relates local contexts via its transitive property.
In particular, if $E\subseteq F \subseteq X$, then, 
\begin{equation}\label{eq:comp_IM}
 M_{E \mid X} \!=\! M_{E \mid F}M_{F\mid X}.
\end{equation}
For two contexts $C,C'\in\mathcal{C}$, this implies $M_{(C\cap C')\mid C}M_{C\mid X}=M_{(C\cap C')\mid C'}M_{C'\mid X}$.

Operationally, a preparation $p$ together with a choice of context $C$ produces a distribution over local sections; we denote it by
\begin{equation}
E_{C\mid p} \in \Delta(O^{C}),
\end{equation}
where $\Delta(O^{C})$ is the simplex of probability distributions on $O^C$.
The family of such outcome distributions for a measurement cover, $\{E_{C\mid p}\}$, is known as an \textit{empirical model}. In this paper, we refer to this family as an empirical \textit{measurement} model to distinguish it from the preparation framework introduced later.
Measurement noncontextuality is then the claim that there exists a single global distribution $D_{X\mid p}\in\Delta(O^{X})$ which respects the structure of the scenario and reproduces all observed context-wise distributions as $E_{C\mid p} = M_{C\mid X} D_{X\mid p}$.
These linear equations are combined into one,
\begin{equation}\label{eq:meas-nc}
    E_p = M_{X} D_{X\mid p},
\end{equation}
where $E$ and $M_{X}$ are the stacks of empirical model and incidence matrices, respectively.

The compositional property of Eq.~\eqref{eq:comp_IM} then implies a \textit{compatibility} constraint on local contexts. 
More specifically, whenever a global section $D_{X\mid p}$ exists, and two contexts overlap, the statistics on the intersection agree.
This can be written in terms of context statistics as,
\begin{equation}\label{eq:no-sig-matrix}
    E_{(C\cap C')\mid p} \;=\; M_{(C\cap C')\mid C}\,E_{C\mid p} \;=\; M_{(C\cap C')\mid C'}\,E_{C'\mid p},
\end{equation}
and it is commonly referred to as the \textit{no-signalling} condition.
It is thus clear that a sheaf-theoretic noncontextual operational theory is non-signalling.
Equation~\eqref{eq:no-sig-matrix} enables a sheaf-theoretic interpretation of the noncontextuality condition in Eq.~\eqref{eq:meas-nc}:
In a noncontextual operational theory, the distributions on local sections can be glued together in a consistent way, yielding a distribution on the global section.

Let us illustrate the above framework with an example, the Abramsky--Brandenburger--Bell model~\cite{Abramsky2011}.
Let the measurements be $X=\{a,a',b,b'\}$, each with outcomes $O=\{0,1\}$.
Suppose the jointly measurable contexts are
\begin{equation}
\mathcal C=\bigl\{\{a,b\},\{a',b\},\{a,b'\},\{a',b'\}\bigr\}.
\end{equation}
Ordering local sections in columns as $(0,0),(1,0),(0,1),(1,1)$, a possible empirical model can be given in the explicit probability table,
\[
\begin{array}{c|cccc}
      & (0,0) & (1,0) & (0,1) & (1,1)\\\hline
(a,b)   & \tfrac12 & 0 & 0 & \tfrac12\\
(a',b)  & \tfrac38 & \tfrac18 & \tfrac18 & \tfrac38\\
(a,b')  & \tfrac38 & \tfrac18 & \tfrac18 & \tfrac38\\
(a',b') & \tfrac18 & \tfrac38 & \tfrac38 & \tfrac18
\end{array}.
\]
It is straightforward to verify that this family is compatible (no-signalling) according to Eq.~\eqref{eq:no-sig-matrix}, and hence a candidate for glueing.

The above empirical model has a quantum realization using two qubits prepared in the Bell state
\begin{equation}
p = \ket{\Phi_+}=\frac{1}{\sqrt2}\bigl(\ket{00}+\ket{11}\bigr),
\end{equation}
followed by  projective measurements in the orthonormal bases,
\begin{equation}
    \begin{aligned}    
        &a=b=\{\ket0,\ket1\}, \\
        & a'=\{\ket{a'_0},\ket{a'_1}\},
        & \ket{a'_0}=c\ket0+s\ket1,
        & \ket{a'_1}=-s\ket0+c\ket1,\\
        & b'=\{\ket{b'_0},\ket{b'_1}\},
        & \ket{b'_0}=c\ket0-s\ket1,
        & \ket{b'_1}=s\ket0+c\ket1
    \end{aligned}
\end{equation}
where $c:=\cos\pi/6$ and $s:=\sin\pi/6$.
For each context $C=\{x,y\}\in\mathcal C$ and $(i,j)\in O^C$, the empirical model is given by the Born rule,
\begin{equation}
E_{C|p}(i,j)=\bra{\Phi_+}\bigl(\Pi^{(x)}_i\otimes \Pi^{(y)}_j\bigr)\ket{\Phi_+},
\end{equation}
where $\Pi^{(x)}_i=\ket{x_i}\!\bra{x_i}$.
However, it was shown that the empirical model of Abramsky--Brandenburger--Bell cannot be consistently extended to a distribution over global sections as in Eq.~\eqref{eq:meas-nc}~\cite{Abramsky2011}, hence it is sheaf-theoretically measurement contextual.
 
\section{Sheaf-theoretic preparation contextuality}

\subsection{Framework}

We now turn to the preparation-dual setting, in which contextuality is probed by varying preparation procedures rather than measurement contexts.
To isolate the structure associated with preparations, we fix a single measurement $m$ with outcome set $O$, and regard all nonclassical features as arising from the manner in which preparation devices can be controlled and combined.

Central to our analysis is a \emph{source}: a preparation device equipped with a finite classical control interface. Choosing a value of this classical control specifies a particular deterministic \emph{preparation instance} of the device.
Preparation instances thus label operationally distinct procedures that can be selected prior to measurement, without making any assumptions about the underlying physical theory describing the prepared system.
Operationally, a preparation instance plays a role dual to that of a measurement outcome: whereas a measurement outcome is the observed value of a random variable in a measurement, a preparation instance labels which preparation procedure is realized.

We write $Y$ for the finite set of sources available in the experiment, and $I$ for the finite set of preparation instances associated with each source.
In general, not all families of sources can be implemented together in a single physical arrangement.
For example, consider the devices that prepare a qubit system in $Z$ and $X$ bases, respectively.
The $Z$ device can prepare the instances $\ketbra{0}$ or $\ketbra{1}$ deterministically, whereas the $X$ device can prepare the instances $\ketbra{-}$ or $\ketbra{+}$.
It is, however, not possible to devise a preparation device that outputs one of $\{\ketbra{0},\ketbra{1},\ketbra{-},\ketbra{+}\}$ deterministically.
Instead, only certain collections of instances admit a joint implementation as alternative preparation instances of a single preparation.
In other words, there are families of sources that can be realized together as coarse-grainings of a single underlying preparation device with a refined classical control.
A trivial example of such a device is the disjoint preparations of a bipartite system as $Z_A$ and $X_B$.
In analogy with measurements, the joint preparability structure is given by a cover $\mathcal S$ of $Y$, whose elements $\Gamma \subseteq Y$ are referred to as \emph{source contexts}.

For a given source context $\Gamma \in \mathcal S$, the experimenter selects a tuple of preparation instances $\sigma_\Gamma \in I^\Gamma$
corresponding to a choice of classical control for each source in the context.
We define this tuple to be the \textit{local preparation section}.
A \textit{global preparation section} is then defined as a tuple $\sigma_Y\in I^Y$.

Upon subsequently performing the measurement $m$, the experimenter observes conditional outcome statistics
\begin{equation}
E_{m \mid \Gamma}(o \mid \sigma_\Gamma),
\qquad
o \in O.
\end{equation}
Sheaf-theoretic preparation contextuality will arise precisely when these local statistics cannot be extended to a single global preparation model ``compatible'' with all source contexts.
We will shortly clarify what compatibility means in our framework.
However, let us first make the above ideas concrete via an example which exhibits sheaf-theoretic preparation contextuality.

As a preparation–dual analogue of the Bell scenario, consider the PBR arrangement~\cite{Pusey2012}.
There are two local laboratories, Alice and Bob, each equipped with two preparation procedures, denoted $a,a'$ for Alice and $b,b'$ for Bob.
Each procedure has two preparation instances labelled $i\in\{0,1\}$.
The joint preparation contexts are defined as,
\begin{equation}
\Gamma \in \bigl\{\{a,b\},\{a,b'\},\{a',b\},\{a',b'\}\bigr\}.
\end{equation}
On each run, one procedure is chosen at each site together with instance labels $(i,j)\in\{0,1\}^2$, and a fixed two–qubit measurement $m$ is performed.

In the quantum realization, Alice and Bob each have two sources, $Z$ and $X$, so that their local preparation instances correspond to single–qubit pure states,
\begin{equation*}
\begin{aligned}
\rho^{a}_0=\rho^{b}_0   &= |0\rangle\langle 0|, &\qquad \rho^{a}_1=\rho^{b}_1   &= |1\rangle\langle 1|,\\
\rho^{a'}_0=\rho^{b'}_0  &= |+\rangle\langle +|, &\qquad \rho^{a'}_1 =\rho^{b'}_1  &= |-\rangle\langle -|,
\end{aligned}
\end{equation*}
where $| \pm \rangle=(|0\rangle\pm|1\rangle)/\sqrt2$.
Thus the joint preparations are product states,
\begin{equation*}
\rho^{\{x,y\}}_{ij} = \rho^x_i \otimes\rho^y_j, \quad \text{for all } x\in\{a,a'\} \text{ and } y\in\{b,b'\}.
\end{equation*}
The PBR measurement is given by the four-outcome POVM $\{M_k\}_k$ whose effects are the rank-one projectors onto the normalized vectors proportional to the vectors,
\begin{align*}
    &|\xi_1\rangle = |01\rangle + |10\rangle, \quad
    |\xi_2\rangle = |0-\rangle + |1+\rangle, \nonumber\\
    &|\xi_3\rangle = |+1\rangle + |-0\rangle, \quad
    |\xi_4\rangle = |+-\rangle + |-+\rangle.
\end{align*}
With these specifications, the empirical model is obtained via $E_{m|\Gamma}(k\mid i,j)=\mathrm{Tr}\!\bigl[M_k\,\rho^{\Gamma}_{ij}\bigr]$ as,
\begin{equation}\label{eq:E-PBR}
    E =
    \left(
    \begin{array}{cccc|cccc|cccc|cccc}
    0 & \tfrac12 & \tfrac12 & 0 &
    \tfrac14 & \tfrac14 & \tfrac14 & \tfrac14 &
    \tfrac14 & \tfrac14 & \tfrac14 & \tfrac14 &
    \tfrac12 & 0 & 0 & \tfrac12 \\
    
    \tfrac14 & \tfrac14 & \tfrac14 & \tfrac14 &
    0 & \tfrac12 & \tfrac12 & 0 &
    \tfrac12 & 0 & 0 & \tfrac12 &
    \tfrac14 & \tfrac14 & \tfrac14 & \tfrac14 \\
    
    \tfrac14 & \tfrac14 & \tfrac14 & \tfrac14 &
    \tfrac12 & 0 & 0 & \tfrac12 &
    0 & \tfrac12 & \tfrac12 & 0 &
    \tfrac14 & \tfrac14 & \tfrac14 & \tfrac14 \\
        
    \tfrac12 & 0 & 0 & \tfrac12 &
    \tfrac14 & \tfrac14 & \tfrac14 & \tfrac14 &
    \tfrac14 & \tfrac14 & \tfrac14 & \tfrac14 &
    0 & \tfrac12 & \tfrac12 & 0
    \end{array}
    \right).
\end{equation}
Here, the columns are ordered in four consecutive blocks corresponding to the contexts \(\{a,b\}\), \(\{a,b'\}\), \(\{a',b\}\), \(\{a',b'\}\), with lexicographic order \((0,0),(0,1),(1,0),(1,1)\) within each block.
Rows are ordered by outcomes $o\in\{1,2,3,4\}$.

In parallel with measurement noncontextuality, we take preparation noncontextuality to mean that there exists a single global response matrix, from which all local empirical preparation models are obtained.
Such a putative global model assigns outcome statistics to every global preparation instance and is described by a stochastic \textit{response matrix},
\begin{equation}
D_{m\mid Y}(o\mid\sigma_Y),
\qquad
\sigma_Y\in I^Y.
\end{equation}
The predicted outcome statistics for the context $\Gamma$ are thus obtained as $E_{m\mid\Gamma}= D_{m\mid Y}S_{Y\mid\Gamma}$. 
We call $S_{Y\mid\Gamma}$, which is a stochastic matrix that plays a role analogous to the incident matrix of Eq.~\eqref{eq:incidence_matrix}, the \emph{extension matrix}.
These linear equations are also combined into one,
\begin{equation}\label{eq:prep-nc}
    E_{m}= D_{m\mid Y}S_{Y}.
\end{equation}
The extension matrix lifts distributions over local preparation sections to distributions over global preparation sections by specifying how the causal structure of preparations is implemented.
Operationally, $S_{Y\mid\Gamma}$ also encodes how preparations for sources outside the context $\Gamma$ are described.

However, there is a sharp contrast between the extension matrix  $S_{Y\mid\Gamma}$ and the incidence matrix  $M_{C\mid X}$.
While the incidence matrix is canonically determined by the physics of the scenario, the extension matrix is not, and in general, there may be many inequivalent ways to complete a partial preparation assignment.
Before asking whether a global response matrix exists, we must therefore specify what compositional properties we expect from the extension matrix.
We impose minimal structural requirements on admissible families of extension matrices, expressing the idea that extension should represent a genuine refinement of preparation control rather than an arbitrary stochastic rewriting.

Suppose $Y$ and $I$ are finite sets of sources and instances, respectively, and let $W\subseteq V\subseteq U\subseteq Y$.
Our first requirement is \emph{input independence}.
This means that the distributions assigned to instances outside a local preparation section do not depend on the instances inside that section.
This also ensures that the empirical preparation model for a given context does not depend on the empirical preparation models of unspecified contexts, so that contextual behaviour cannot be artificially built into the completion rule itself.
This is arguably a reasonable minimal assumption, as it allows us to define what a context is in the first place.
\begin{definition}[Input independence]\label{def:input-independence}
An extension matrix $S_{U\mid V}$ is \emph{input independent} if there exists a distribution
\begin{equation}
\mu_{U\setminus V}\in\Delta(I^{U\setminus V})
\end{equation}
such that for all $\sigma_V\in I^V$ and $\sigma_U\in I^U$,
\begin{equation}\label{eq:ctx-ind-form}
S_{U\mid V}(\sigma_U\mid\sigma_V)
=
\mathbf 1[(\sigma_U)|_V=\sigma_V]\;
\mu_{U\setminus V}((\sigma_U)|_{U\setminus V}).
\end{equation}
\end{definition}

Our second requirement is \emph{compositionality}, which means extending in stages yields the same distribution as extending in a single step, ensuring that refinement is independent of arbitrary choices of intermediate implementation.

\begin{definition}[Compositionality]\label{def:compositionality}
A family $\{S_{U\mid V}\}$ is \emph{compositional} if for all $W\subseteq V\subseteq U$,
\begin{equation}\label{eq:compositionality-app}
    S_{U\mid W}=S_{U\mid V}\,S_{V\mid W}.
\end{equation}
\end{definition}
\noindent The compositionality above is parallel to that of incidence matrices in Eq.~\eqref{eq:comp_IM}; see Appendix~\ref{app:stf} for a formal interpretation of this structure.
Together, these conditions express that the stochastic extension behaves as a coherent refinement procedure, independent of staging and of the values of previously fixed preparation instances.
We refer to a family of extension matrices $\{S_{U|V}\}$ satisfying input independence and compositionality as an \emph{admissible extension family}.
These structural requirements rigidly constrain the form of admissible extension matrices.

\begin{theorem}[Product-form extension theorem]
\label{thm:product-prior}
Assume $\{S_{U\mid V}\}$ is input independent and compositional.
Then there exist single-site distributions $\{\mu_p\in\Delta(I)\}_{p\in Y}$ such that for all $V\subseteq U\subseteq Y$,
\begin{equation}\label{eq:product-prior}
    S_{U\mid V}(\sigma_U\mid\sigma_V) = \mathbf 1[(\sigma_U)|_V=\sigma_V]\, \prod_{p\in U\setminus V} \mu_p(\sigma_p).
\end{equation}
In particular, for extensions from context $\Gamma\in\mathcal{S}$ to $Y$,
\begin{equation}\label{eq:product-prior2}
    S_{Y\mid\Gamma}(\sigma_Y\mid\sigma_\Gamma) = \mathbf 1[(\sigma_Y)|_\Gamma=\sigma_\Gamma]\, \prod_{p\in Y\setminus\Gamma} \mu_p(\sigma_p).
\end{equation}
\end{theorem}

\begin{proof}
Fix $W\subseteq V\subseteq U$. Using Eq.~\eqref{eq:ctx-ind-form} in the compositionality condition Eq.~\eqref{eq:compositionality-app} gives, for all $\sigma_W\in I^W$ and $\sigma_U\in I^U$,
\begin{equation}
    \begin{multlined}
        \mathbf 1[(\sigma_U)|_W=\sigma_W]\; \mu_{U\setminus W}(\sigma_U|_{U\setminus W}) \\
        = \sum_{\sigma_V\in I^V} \mathbf 1[(\sigma_U)|_V=\sigma_V]\;
        \mathbf 1[(\sigma_V)|_W=\sigma_W] \\
        {}\times \mu_{U\setminus V}(\sigma_U|_{U\setminus V})\;
        \mu_{V\setminus W}(\sigma_V|_{V\setminus W}) .
    \end{multlined}
\end{equation}
The indicator functions force $\sigma_V=\sigma_U|_V$, so the sum collapses to
\begin{equation}
    \begin{aligned}
        &\mathbf 1[(\sigma_U)|_W=\sigma_W]\; \mu_{U\setminus W}((\sigma_U)|_{U\setminus W}) \\
        &\qquad
        =
        \mathbf 1[(\sigma_U)|_W=\sigma_W]\; \mu_{U\setminus V}((\sigma_U)|_{U\setminus V})\; \mu_{V\setminus W}((\sigma_U)|_{V\setminus W}),
    \end{aligned}
\end{equation}
and hence
\begin{equation}\label{eq:multiplicative}
    \mu_{U\setminus W}(\sigma_{U\setminus V},\sigma_{V\setminus W})
    =
    \mu_{U\setminus V}(\sigma_{U\setminus V})\; \mu_{V\setminus W}(\sigma_{V\setminus W}),
\end{equation}
for all $\sigma_{U\setminus V}\in I^{U\setminus V}$ and $\sigma_{V\setminus W}\in I^{V\setminus W}$.

Now fix $V\subseteq U$ and choose any chain
\begin{equation}
V=V_0\subset V_1\subset\cdots\subset V_n=U,
\qquad
V_{i+1}=V_i\cup\{p_{i+1}\}.
\end{equation}
Repeated application of Eq.~\eqref{eq:multiplicative} yields
\begin{equation}
\mu_{U\setminus V}(\sigma_{U\setminus V})
=
\prod_{i=1}^n \mu_{V_i\setminus V_{i-1}}(\sigma_{p_i}).
\end{equation}
Defining $\mu_{p_i}:=\mu_{V_i\setminus V_{i-1}}\in\Delta(I)$, this factorisation is independent of the chosen chain, since \eqref{eq:multiplicative} implies $\mu_{A\cup B}=\mu_A\,\mu_B$ for disjoint $A,B$.
Substituting this product form back into \eqref{eq:ctx-ind-form} establishes \eqref{eq:product-prior}.
\end{proof}

Given an empirical preparation model $\{E_{m\mid\Gamma}\}_{\Gamma\in\mathcal S}$, we are interested in admissible extension families that induce consistent restrictions on overlaps of source contexts.
Operationally, this expresses preparation compatibility: the effective description on $\Gamma\cap\Gamma'$ should not depend on whether it is obtained via $\Gamma$ or via $\Gamma'$.
\begin{definition}[Preparation compatibility]
We say the empirical preparation model $\{E_{m\mid\Gamma}\}_{\Gamma\in\mathcal S}$ is \emph{preparation compatible} if and only if there exists an admissible family of extension matrices $\{S_{U\mid V}\}$ such that for all $\Gamma,\Gamma'\in\mathcal S$,
\begin{equation}\label{eq:prep-comp}
E_{m\mid\Gamma}\,S_{\Gamma\mid \Gamma\cap\Gamma'}
=
E_{m\mid\Gamma'}\,S_{\Gamma'\mid \Gamma\cap\Gamma'}.
\end{equation}
\end{definition}

Equation~\eqref{eq:prep-comp} states that the induced outcome statistics on the shared sources $\Gamma\cap\Gamma'$ are independent of which larger source context is used.
In the sheaf-theoretic measurement framework, so-called generalized no-signalling expresses the compatibility of probability distributions on overlaps of measurement contexts~\cite{Abramsky2011}. In Bell scenarios this coincides with causal no-signalling, but in general it is a structural consistency condition on marginal distributions rather than a dynamical constraint. Eq.~\eqref{eq:prep-comp} plays the analogous mathematical role for preparation scenarios: it enforces compatibility of the induced statistics on overlapping source contexts, now mediated by extension matrices rather than marginalisation. Unlike the measurement case, however, this condition does not correspond to a physical no-signalling requirement, but simply expresses coherence between the empirical model and the chosen extension structure.

With these ingredients in place, we can now state the preparation-dual notion of contextuality.
The central question is whether the observed statistics on each source context can be reproduced by a single global response matrix for some admissible extension structure.

\begin{definition}[Preparation contextuality]\label{def:prep-cont}
An empirical preparation model $\{E_{m\mid\Gamma}\}_{\Gamma\in\mathcal S}$ is 
\emph{preparation noncontextual} if and only if there exists an admissible family  of stochastic extension matrices $\{S_{U\mid V}\}$ and a global response matrix $D_{m\mid Y}$ such that
\begin{equation}
E_{m}
=
D_{m\mid Y}\,S_{Y},
\end{equation}
otherwise it is preparation contextual.
\end{definition}
Although Definition~\ref{def:prep-cont} applies to arbitrary empirical preparation models, the substantive cases are the preparation compatible ones. If no admissible extension family satisfies Eq.~\eqref{eq:prep-comp}, contextuality is already witnessed at the level of compatibility. Conversely, any factorisation $E_m = D_{m\mid Y}S_Y$ automatically yields compatibility for the corresponding extension family, as shown in Lemma~\ref{lem:Fact-PrepComp}. Unlike in the measurement case, where compatibility is fixed by canonical marginalisation, a preparation model may be compatible with respect to multiple admissible extension families, or none.

\begin{lemma}
\label{lem:Fact-PrepComp}
Let $\{E_{m|\Gamma}\}_{\Gamma\in\mathcal S}$ be an empirical preparation model with an admissible family of extension matrices $\{S_{U|V}\}$ and a global response matrix $D_{m|Y}$ such that, for every
$\Gamma\in\mathcal S$,
\begin{equation}
E_{m|\Gamma}=D_{m|Y}S_{Y|\Gamma}.
\end{equation}
Then the empirical preparation model is preparation compatible.
In particular, for all $\Gamma,\Gamma'\in\mathcal S$,
\begin{equation}
E_{m|\Gamma}S_{\Gamma|\Gamma\cap\Gamma'}
=
E_{m|\Gamma'}S_{\Gamma'|\Gamma\cap\Gamma'} .
\end{equation}
\end{lemma}

\begin{proof}
Using the assumed global representation and compositionality of the extension matrices, we have
\begin{equation}
E_{m|\Gamma}S_{\Gamma|\Gamma\cap\Gamma'}
=
D_{m|Y}S_{Y|\Gamma}S_{\Gamma|\Gamma\cap\Gamma'}
=
D_{m|Y}S_{Y|\Gamma\cap\Gamma'}.
\end{equation}
Similarly,
\begin{equation}
E_{m|\Gamma'}S_{\Gamma'|\Gamma\cap\Gamma'}
=
D_{m|Y}S_{Y|\Gamma'}S_{\Gamma'|\Gamma\cap\Gamma'}
=
D_{m|Y}S_{Y|\Gamma\cap\Gamma'}.
\end{equation}
Hence
\begin{equation}
E_{m|\Gamma}S_{\Gamma|\Gamma\cap\Gamma'}
=
E_{m|\Gamma'}S_{\Gamma'|\Gamma\cap\Gamma'} ,
\end{equation}
as required.
\end{proof}

It is often useful to consider a stronger, possibilistic notion of contextuality.
Given an empirical model $E$, its possibilistic reduction $\overline E$ is obtained by mapping all nonzero probabilities to one retaining only the support of the empirical statistics.
In turn, stochastic matrices and matrix multiplication are replaced by Boolean matrices and Boolean matrix multiplication, respectively. 

\begin{theorem}[Possibilistic contextuality implies probabilistic contextuality]\label{thm:poss-prob}
If an empirical model is possibilistically preparation contextual, then it is also preparation contextual in the probabilistic sense.
\end{theorem}

\begin{proof}
Any probabilistic global response matrix induces a possibilistic one by taking supports: replacing each positive entry by $1$ and each zero by $0$. Hence the existence of a probabilistic noncontextual model would imply the existence of a possibilistic one.
The contrapositive yields the claim.
\end{proof}

\subsection{Example: PBR-type preparation scenario}
We now demonstrate that the PBR-type preparation scenario constructed earlier is preparation contextual.
Here we outline the proof and refer the interested reader to  Appendix~\ref{app:pbr-proof} for the details.

Recall from Eq.~\eqref{eq:E-PBR} that this empirical model $E$ takes the explicit form,
\begin{equation}
E =
\left(
\begin{array}{cccc|cccc|cccc|cccc}
0 & \tfrac12 & \tfrac12 & 0 &
\tfrac14 & \tfrac14 & \tfrac14 & \tfrac14 &
\tfrac14 & \tfrac14 & \tfrac14 & \tfrac14 &
\tfrac12 & 0 & 0 & \tfrac12 \\

\tfrac14 & \tfrac14 & \tfrac14 & \tfrac14 &
0 & \tfrac12 & \tfrac12 & 0 &
\tfrac12 & 0 & 0 & \tfrac12 &
\tfrac14 & \tfrac14 & \tfrac14 & \tfrac14 \\

\tfrac14 & \tfrac14 & \tfrac14 & \tfrac14 &
\tfrac12 & 0 & 0 & \tfrac12 &
0 & \tfrac12 & \tfrac12 & 0 &
\tfrac14 & \tfrac14 & \tfrac14 & \tfrac14 \\

\tfrac12 & 0 & 0 & \tfrac12 &
\tfrac14 & \tfrac14 & \tfrac14 & \tfrac14 &
\tfrac14 & \tfrac14 & \tfrac14 & \tfrac14 &
0 & \tfrac12 & \tfrac12 & 0
\end{array}
\right).
\end{equation}
By Theorem~\ref{thm:product-prior}, any admissible family of extension matrices has product form of Eq.~\eqref{eq:product-prior2} for some choice of single-site distributions $\{\mu_p\}_{p\in Y}$.
It turns out that taking each $\mu_p$ to be uniform yields an admissible extension structure that satisfies Eq.~\eqref{eq:prep-comp}.
Thus, the scenario is preparation compatible, and hence the question of preparation contextuality is nontrivial.
It remains to show that no global response matrix $D_{m\mid Y}$ can reproduce $E$ via Eq.~\eqref{eq:prep-nc} for any admissible extension matrix.

Assume, for contradiction, that there exists a global response matrix $D_{m\mid Y}$ together with a product-form extension matrix reproducing the empirical statistics via Eq.~\eqref{eq:prep-nc}.
Then, note that for each preparation instance in the model, there exists exactly one impossible outcome.
Because the extension matrix assigns a nonzero weight to every globally consistent refinement of a local preparation choice, each forbidden outcome at the empirical level propagates to a forbidden outcome at the level of the corresponding global preparation instances.
Hence, the pattern of zeros in the empirical matrices imposes stringent support constraints on the columns of $D_{m\mid Y}$.

The crucial observation is that these constraints are incompatible with the requirement that every column of $D_{m\mid Y}$ represent a valid probability distribution.
A parity-based combinatorial argument, as detailed in Appendix~\ref{app:pbr-proof}, shows that for any choice of single-site extension distributions $\{\mu_p\}$, at least one global preparation instance is forced to forbid all four measurement outcomes simultaneously, which is impossible.
Hence, no global response matrix exists for any admissible extension family, and the empirical model is preparation contextual.

Conceptually, the argument proceeds in three steps.
First, the presence of a single forbidden outcome in each empirical column allows one to associate with every local preparation choice a unique outcome that cannot occur.
Second, the product-form extension matrix ensures that whenever a global preparation instance is compatible with a given local choice, the corresponding forbidden outcome is also forbidden globally.
This propagates the empirical zero constraints to the global response matrix.

Finally, one examines the pattern of propagated constraints across the four source contexts.
For global preparation instances of even parity, the forbidden outcomes arising from the four contexts are all distinct, so that all four measurement outcomes are impossible simultaneously.
For odd parity, fewer outcomes are forbidden, but the remaining support is still insufficient to satisfy consistency with all empirical columns.
In both cases, one obtains a contradiction, showing that no global response matrix can exist.

\section{Discussion and Outlook}

We have introduced a preparation-dual notion of contextuality formulated as an obstruction to stochastic extension.
In direct analogy with the sheaf-theoretic formulation of measurement contextuality~\cite{Abramsky2011}, preparation contextuality arises when locally specified statistics cannot be extended to a single global response matrix compatible with all source contexts.

To convey the intuition underlying our framework, we draw a conceptual parallel with sheaf-theoretic measurement contextuality.
The framework of sheaf-theoretic measurement contextuality rests on the following noncontextuality assumption: There is a ``Platonic'' set of simultaneous global outcome assignments, and each local outcome (section) stems from one such global assignment through a restriction map so that locality is obtained by forgetting the unobserved events.
The restrictions have linear-algebraic representations via incidence matrices.
Importantly, the local outcomes are observed noncontextually, i.e., they do not depend on the contexts in which they were obtained.
In the sheaf-theoretic picture, the latter is captured by the compatibility property of sections, which manifests in the transitivity of the incidence matrices.
This formulation places sheaf-theoretic contextuality within a linear-algebraic framework.
With the noncontextuality assumption built into the incidence matrix, a family of local data (an empirical measurement model) is measurement noncontextual if it can be recovered from a single global section via these matrices.

In our preparation contextuality framework, we assumed Platonic simultaneous global instances.
An extension matrix then lifts distributions induced by instances within contexts to global distributions induced by global instances where some preparation instances are left free. 
We thus needed to encode a noncontextuality assumption into the extension matrices by consistently completing the unspecified instance distributions. 
We achieved this by requiring the following, inspired by the measurement case: (i) that local distributions be independent of the contexts in which they are assigned (input independence), thereby excluding spurious correlations between independently specified sources; and (ii) that the extension matrices be transitive (compositionality), thereby enforcing consistency under successive extensions.
These conditions ensure that extensions behave as a coherent refinement procedure. 
The key distinction, however, is that the extension matrices are not canonical, owing to their stochastic nature and the freedom to choose different distributions for preparation instances.
With extension matrices encoding the noncontextuality assumption, a family of local data (an empirical preparation model) is preparation noncontextual if it can be recovered from these matrices via a global response matrix.

We have shown that, by exploiting a linear-algebraic perspective, the sheaf-theoretic notion of contextuality can be meaningfully extended to preparations, which
captures a structural incompatibility between local preparation descriptions and global response models.
This parallel enabled us to leverage the possibilistic contextuality of the measurement scenario in the case of preparations.
A natural question is thus the extent to which this parallelism can be sustained.
For example, a deeper structural question concerns the existence of a preparation-analogue of Fine’s theorem~\cite{Fine1982}.
In the measurement case, Abramsky and Brandenburger showed that the existence of a global section is equivalent to the existence of a deterministic hidden-variable model~\cite{Abramsky2011}.
This result provides a precise bridge between the sheaf-theoretic and ontological-model formulations of contextuality.
An analogous result would clarify the exact relationship between the sheaf-theoretic preparation contextuality introduced here and generalized preparation contextuality~\cite{Spekkens2005}.

More broadly, the extension-based perspective suggests several avenues for further investigation.
First, one may seek a cohomological characterisation of preparation contextuality paralleling developments for measurement contextuality~\cite{AbramskyMansfield2011}.
Second, the linear-algebraic structure of the obstruction invites algorithmic analysis:
preparation noncontextuality reduces to a feasibility problem for stochastic matrices subject to linear constraints.
Third, it would be natural to investigate whether a unified categorical or sheaf-theoretic treatment of preparation and measurement contextuality could illuminate deeper dualities between restriction and extension.

\bibliography{arXiv-quantum}
\bibliographystyle{apsrev4-1}

\begin{widetext}

\appendix

\setcounter{equation}{0}
\renewcommand{\theequation}{\thesection.\arabic{equation}}

\newpage

\section{Sheaf-Theoretic Formulation of Preparation Contextuality}
\label{app:stf}

In this appendix we recast the preparation framework developed in the main text in the language of sheaf theory, following the general structure of Abramsky and Brandenburger~\cite{Abramsky2011}, but adapting it to the preparation-dual setting. The key distinction is that, whereas measurement contextuality is governed by canonical restriction maps, the preparation setting requires stochastic extension maps as additional structure.

\medskip

We begin by recalling the measurement case. Let $X$ be a finite set of measurements with outcome set $O$, and let $\mathcal{C} \subseteq \mathcal{P}(X)$ be a measurement cover. The event assignment is given by a presheaf
\[
\mathcal{E} : \mathcal{P}(X)^{\mathrm{op}} \to \mathbf{Set}, 
\qquad U \mapsto O^{U},
\]
with restriction maps $\mathrm{res}^U_V := \mathcal{E}(V \subseteq U)$ given by projection. Composing with the distribution functor over a commutative semiring $R$, one obtains a presheaf
\[
\mathcal{D}_R \mathcal{E} : \mathcal{P}(X)^{\mathrm{op}} \to \mathbf{Set}, 
\qquad U \mapsto \mathcal{D}_R(O^{U}),
\]
whose morphisms are given by marginalisation. In the probabilistic setting one takes $R=\mathbb{R}_{\ge 0}$, while $R=\mathbb{B}$ yields the possibilistic case.

An empirical model is a family $\{ e_C \in \mathcal{D}_R \mathcal{E}(C) \}_{C \in \mathcal{C}}$ satisfying compatibility on overlaps,
\[
\mathcal{D}_R(\mathrm{res}^{C}_{C \cap C'})(e_C)
=
\mathcal{D}_R(\mathrm{res}^{C'}_{C \cap C'})(e_{C'}).
\]
Measurement noncontextuality is equivalent to the existence of a global section $d \in \mathcal{D}_R \mathcal{E}(X)$ such that $\mathrm{res}^X_C(d) = e_C$ for all $C \in \mathcal{C}$.

\medskip

We now turn to the preparation scenario. Let $Y$ be a finite set of sources, each with instance set $I$, and let $\mathcal{S} \subseteq \mathcal{P}(Y)$ be a family of source contexts. We define the instance presheaf
\[
\mathcal{I} : \mathcal{P}(Y)^{\mathrm{op}} \to \mathbf{Set}, 
\qquad U \mapsto I^{U},
\]
with restriction maps $\mathrm{res}^U_V := \mathcal{I}(V \subseteq U)$.

\medskip

To capture the preparation structure, we introduce stochastic extension maps. These are naturally described using the Kleisli category of the distribution monad $\mathcal{D}_R$~\cite{Jacobs2015}. Morphisms in this category are stochastic maps
\[
f : A \to \mathcal{D}_R(B),
\]
which may be written as conditional distributions $f(b \mid a)$. Given $f : A \to \mathcal{D}_R(B)$ and $g : B \to \mathcal{D}_R(C)$, their composition is defined by
\[
(g \star f)(c \mid a)
=
\sum_{b \in B} g(c \mid b)\, f(b \mid a).
\]

An extension structure is defined as a covariant functor
\[
\mathcal{S} : \mathcal{P}(Y) \to \mathbf{Kl}(\mathcal{D}_R), 
\qquad U \mapsto I^{U},
\]
assigning to each inclusion $V \subseteq U$ a stochastic map
\[
\mathrm{ext}^V_U := \mathcal{S}(V \subseteq U) : I^V \to \mathcal{D}_R(I^U),
\]
such that for all $W \subseteq V \subseteq U$,
\[
\mathrm{ext}^W_U
=
\mathrm{ext}^V_U \star \mathrm{ext}^W_V.
\]
This is the \textbf{compositionality} condition.

\medskip

We say that the extension structure is \textbf{input-independent} if, for each $V \subseteq U$, there exists a distribution $\mu_{U \setminus V} \in \mathcal{D}_R(I^{U \setminus V})$ such that
\[
\mathrm{ext}^V_U(\sigma_U \mid \sigma_V)
=
\mathbf{1}\!\left[(\sigma_U)|_V = \sigma_V\right]\,
\mu_{U \setminus V}((\sigma_U)|_{U \setminus V}).
\]
In particular, this implies that extending and then restricting recovers the original assignment, i.e.
\[
\mathcal{D}_R(\mathrm{res}^U_V)\circ \mathrm{ext}^V_U = \delta_{I^V}.
\]

\medskip

An empirical preparation model assigns to each source context $\Gamma \in \mathcal{S}$ a map
\[
e^{\Gamma} : I^{\Gamma} \to \mathcal{D}_R(O).
\]
We say that the empirical model is \textbf{preparation compatible} if there exists an admissible extension structure $\mathcal S$ such that, for all $\Gamma,\Gamma'\in\mathcal S$,
\[
e^\Gamma \star \mathrm{ext}^{\Gamma\cap\Gamma'}_\Gamma
=
e^{\Gamma'} \star \mathrm{ext}^{\Gamma\cap\Gamma'}_{\Gamma'} .
\]

\medskip

A global section is a map
\[
d : I^{Y} \to \mathcal{D}_R(O).
\]
We say that the empirical model is \textbf{preparation noncontextual} if there exists an admissible extension structure $\mathcal{S}$ and a global section $d$ such that
\[
e^{\Gamma}
=
d \star \mathrm{ext}^{\Gamma}_{Y},
\quad \Gamma \in \mathcal{S}.
\]
If no such global section exists, the empirical model is preparation contextual; in the preparation compatible case, this gives a nontrivial obstruction beyond failure of compatibility.

\medskip

Finally, under compositionality and input independence, the extension maps necessarily take the product form established in Theorem~\ref{thm:product-prior}:
\[
\mathrm{ext}^V_U(\sigma_U \mid \sigma_V)
=
\mathbf{1}\!\left[(\sigma_U)|_V = \sigma_V\right]
\prod_{p \in U \setminus V} \mu_p(\sigma_p).
\]

\section{PBR-type example: full proof}
\label{app:pbr-proof}

In this appendix we give a complete technical proof that the quantum model discussed in the main text is preparation contextual in the sense of Definition~4.

We first exhibit an admissible extension family satisfying preparation compatibility, showing that the empirical model admits a coherent extension structure and that the question of preparation contextuality is nontrivial. This is achieved by taking the single-site distributions $\{\mu_p\}_{p\in Y}$ to be uniform.

It remains to show that no admissible extension family admits a global response matrix reproducing the empirical model. We in fact prove the stronger statement that no \emph{possibilistic} global response matrix exists. Preparation contextuality in the probabilistic sense then follows from Theorem~2.

We work with the set of sources
\[
Y=\{a,b,a',b'\},
\]
each admitting binary instances \(I=\{0,1\}\). The implementable source contexts are
\[
\mathcal S=
\bigl\{
\{a,b\},\{a,b'\},\{a',b\},\{a',b'\}
\bigr\}.
\]
A global instance assignment is a tuple
\[
\sigma_Y=(a,b,a',b')\in I^4.
\]

By Theorem~1, any admissible family of extension matrices has product form. Thus, for each \(p\in Y\), fix a single-site distribution \(\mu_p\in\Delta(I)\), and write
\[
S_{Y\mid \Gamma}(\sigma_Y\mid \sigma_\Gamma)
=
\mathbf 1[(\sigma_Y)|_\Gamma=\sigma_\Gamma]
\prod_{p\in Y\setminus\Gamma}\mu_p(\sigma_p).
\]
The special case in which every \(\mu_p\) is uniform gives an admissible extension family. A direct computation using the empirical model below shows that this uniform extension family satisfies Eq.~(25). Hence the empirical model is preparation compatible.

It remains to show that no admissible extension family admits a global response matrix. Since every admissible extension family has product form, it suffices to rule out global response matrices for arbitrary choices of the single-site distributions \(\{\mu_p\}_{p\in Y}\).

Columns of empirical matrices are ordered in four consecutive blocks corresponding to the contexts
\[
\{a,b\},\quad \{a,b'\},\quad \{a',b\},\quad \{a',b'\},
\]
with lexicographic order \((0,0),(0,1),(1,0),(1,1)\) within each block. Rows are ordered by outcomes \(o\in\{1,2,3,4\}\), and rows of extension and global response matrices are ordered lexicographically by \(\sigma_Y\in I^4\).

With these conventions, the empirical model \(E\) takes the explicit form
\[
E =
\left(
\begin{array}{cccc|cccc|cccc|cccc}
0 & \tfrac12 & \tfrac12 & 0 &
\tfrac14 & \tfrac14 & \tfrac14 & \tfrac14 &
\tfrac14 & \tfrac14 & \tfrac14 & \tfrac14 &
\tfrac12 & 0 & 0 & \tfrac12 \\

\tfrac14 & \tfrac14 & \tfrac14 & \tfrac14 &
0 & \tfrac12 & \tfrac12 & 0 &
\tfrac12 & 0 & 0 & \tfrac12 &
\tfrac14 & \tfrac14 & \tfrac14 & \tfrac14 \\

\tfrac14 & \tfrac14 & \tfrac14 & \tfrac14 &
\tfrac12 & 0 & 0 & \tfrac12 &
0 & \tfrac12 & \tfrac12 & 0 &
\tfrac14 & \tfrac14 & \tfrac14 & \tfrac14 \\

\tfrac12 & 0 & 0 & \tfrac12 &
\tfrac14 & \tfrac14 & \tfrac14 & \tfrac14 &
\tfrac14 & \tfrac14 & \tfrac14 & \tfrac14 &
0 & \tfrac12 & \tfrac12 & 0
\end{array}
\right).
\]
Each column contains exactly one zero entry. Let \(\overline E\) denote the possibilistic reduction of \(E\), obtained by mapping \(0\mapsto 0\) and \((0,1]\mapsto 1\). Let \(\overline S_Y\) denote the possibilistic reduction of the stacked extension matrix \(S_Y\).

Since every column of \(\overline E\) has a unique zero entry, there exists a unique function
\[
\varphi:\{1,\dots,16\}\to\{1,2,3,4\}
\]
such that
\[
\overline E(\varphi(j),j)=0
\]
for each column \(j\).

\begin{lemma}
If \(\overline E=\overline D\,\overline S_Y\), then
\[
\overline S_Y(k,j)=1
\quad\Longrightarrow\quad
\overline D(\varphi(j),k)=0.
\]
\end{lemma}

\begin{proof}
If \(\overline S_Y(k,j)=1\) and \(\overline E(i,j)=0\), Boolean matrix multiplication forces \(\overline D(i,k)=0\). Taking \(i=\varphi(j)\) gives the claim.
\end{proof}

For each global assignment index \(k\), define
\[
J(k):=\{j:\overline S_Y(k,j)=1\}.
\]

\begin{lemma}
If \(\overline E=\overline D\,\overline S_Y\), then
\[
\operatorname{supp}(\overline D(:,k))
\subseteq
\{1,2,3,4\}\setminus \varphi(J(k)),
\]
and hence
\[
|\operatorname{supp}(\overline D(:,k))|
\le
4-|\varphi(J(k))|.
\]
\end{lemma}

\begin{proof}
Each \(j\in J(k)\) forbids the outcome \(\varphi(j)\) in column \(k\) by the previous lemma.
\end{proof}

The required combinatorial fact is the following.

\begin{lemma}
Suppose \(J(k)\) contains one column from each context block, and let the corresponding global assignment be
\[
(a,b,a',b')\in I^4.
\]
If
\[
a\oplus b\oplus a'\oplus b'=0,
\]
then
\[
\varphi(J(k))=\{1,2,3,4\}.
\]
If
\[
a\oplus b\oplus a'\oplus b'=1,
\]
then
\[
|\varphi(J(k))|=2.
\]
\end{lemma}

\begin{proof}
This follows by direct inspection of the zero pattern of the empirical model \(E\).
\end{proof}

\begin{theorem}
There exists no possibilistic matrix \(\overline D\) with nonempty columns such that
\[
\overline E=\overline D\,\overline S_Y
\]
for any choice of single-site distributions \(\{\mu_p\}_{p\in Y}\).
\end{theorem}

\begin{proof}
Assume, for contradiction, that such a \(\overline D\) exists.

Choose any
\[
(a,b,a',b')
\in
\operatorname{supp}(\mu_a)\times
\operatorname{supp}(\mu_b)\times
\operatorname{supp}(\mu_{a'})\times
\operatorname{supp}(\mu_{b'}).
\]
Let \(k\) be the corresponding global assignment index. Since each selected value lies in the support of the corresponding single-site distribution, the compatible completion of each local context has nonzero weight. Hence \(J(k)\) contains one column from each of the four context blocks.

If the tuple \((a,b,a',b')\) has even parity, the previous lemma gives
\[
\varphi(J(k))=\{1,2,3,4\}.
\]
Therefore
\[
\operatorname{supp}(\overline D(:,k))=\varnothing,
\]
contradicting the requirement that columns of \(\overline D\) be nonempty.

It remains to consider the case in which every tuple in
\[
\operatorname{supp}(\mu_a)\times
\operatorname{supp}(\mu_b)\times
\operatorname{supp}(\mu_{a'})\times
\operatorname{supp}(\mu_{b'})
\]
has odd parity. This can happen only if each \(\operatorname{supp}(\mu_p)\) is a singleton; otherwise, changing one supported bit while keeping the others fixed would produce a tuple of opposite parity. Hence there is a unique global assignment \(k_\ast\) with nonzero support.

For this assignment, the parity lemma gives
\[
|\varphi(J(k_\ast))|=2.
\]
Thus
\[
|\operatorname{supp}(\overline D(:,k_\ast))|\le 2.
\]
However, since \(k_\ast\) is the only supported global assignment, Boolean consistency with any empirical column \(j\in J(k_\ast)\) requires
\[
\overline E(:,j)=\overline D(:,k_\ast).
\]
Each such empirical column has support of size \(3\), since every column of \(E\) contains exactly one zero. Hence
\[
|\operatorname{supp}(\overline D(:,k_\ast))|\ge 3,
\]
a contradiction.
\end{proof}

Thus no possibilistic global response matrix exists for any admissible extension family. By Theorem~\ref{thm:poss-prob}, no probabilistic global response matrix exists either. Since compatibility was established above, this gives a nontrivial instance of preparation contextuality.

\end{widetext}

\end{document}